\documentclass[letterpaper, 10 pt, conference]{ieeeconf}  % Comment this line out if you need a4paper

\IEEEoverridecommandlockouts                              % This command is only needed if 
\renewcommand{\baselinestretch}{0.965}
\usepackage{amsmath,amsfonts}
\usepackage{algorithmic}
\usepackage{algorithm}
\usepackage{array}
\usepackage[caption=false,font=normalsize,labelfont=sf,textfont=sf]{subfig}
\usepackage{textcomp}
\usepackage{stfloats}
\usepackage{url}
\usepackage{verbatim}
\usepackage{graphicx}
\usepackage{cite}
\usepackage{multirow}
\usepackage{amssymb}
\usepackage{mathtools}
\usepackage[mathscr]{euscript}
\usepackage{booktabs}
\usepackage{tabularx,booktabs}
\newcolumntype{C}{>{\centering\arraybackslash}X} 
\usepackage{lipsum} % filler text

\DeclareMathOperator{\EX}{\mathbb{E}}% expected value
\newtheorem{exmp}{Example}[section]
\newtheorem{theorem}{Theorem}[section]

\usepackage{enumerate}
\usepackage{tikz}
\usetikzlibrary{shapes,arrows,positioning}
\usetikzlibrary{shapes.geometric, arrows}
\usepackage{tabularx}
\usepackage{caption}
\title{\LARGE \bf
Integrated User Matching and Pricing in Round-Trip Car-Sharing
}

\author{Avalpreet Singh Brar$^{1}$, Rong Su$^{2}$,  Gioele Zardini$^{3}$, Jaskaranveer Kaur$^{1}$  % <-this % stops a space
\thanks{This study is supported under the RIE2020 Industry Alignment Fund – Industry Collaboration Projects (IAF-ICP) Funding Initiative, as well as cash and in-kind contribution from the industry partner(s).}%
\thanks{$^{1}$Continental Automotive, Singapore ({ \{avalpreet.singh.brar, jaskaranveer.kaur\}@continental-corporation.com}).}
\thanks{$^{2}$ School of Electrical and Electronic Engineering, Nanyang Technological University, Singapore. (rsu@ntu.edu.sg).}
\thanks{$^{3}$Laboratory for Information and Decision Systems, Massachusetts Institute of Technology, Cambridge, MA, USA. (gzardini@mit.edu).}
}

\begin{document}

\maketitle
\thispagestyle{empty}
\pagestyle{empty}

%%%%%%%%%%%%%%%%%%%%%%%%%%%%%%%%%%%%%%%%%%%%%%%%%%%%%%%%%%%%%%%%%%%%%%%%%%%%%%%%
\begin{abstract}
Traditional round-trip car rental systems mandate users to return vehicles to their point of origin, limiting the system’s adaptability to meet diverse mobility demands. 
This constraint often leads to fleet under-utilization and incurs high parking costs for idle vehicles. 
To address this inefficiency, we propose a N-user matching algorithm which is designed to facilitate one-way trips within the round-trip rental framework. Our algorithm addresses the joint problem of optimal pricing and user matching through a Two-Stage Integer Linear Programming (ILP)-based formulation. 
In the first stage, optimal rental prices are determined by setting a risk factor that governs the likelihood of matching a set of N-user. The second stage involves maximizing expected profit through a novel ILP-based user-matching formulation. Testing our algorithm on real-world scenarios demonstrates an approximate 35\% increase in demand fulfillment. Additionally, we assess the model’s robustness under uncertainty by varying factors such as the risk factor (probability of user’s ride acceptance at the offered price), cost factor (rental cost-to-fare ratio), and maximum chain length.
\end{abstract}

%%%%%%%%%%%%%%%%%%%%%%%%%%%%%%%%%%%%%%%%%%%%%%%%%%%%%%%%%%%%%%%%%%%%%%%%%%%%%%%%
\section{INTRODUCTION}
The round-trip car-sharing model, which is prevalent within the car-sharing industry, requires that users return the vehicles to the original rental station.
%Round-trip car-sharing is a prevalent model within the car-sharing industry, where users must return the vehicles to the station from where the vehicle was rented.
This requirement significantly restricts the service's ability to meet the overall rental demand, often leading to fleet under-utilization and high parking costs for idle vehicles~\cite{le2014new}.
In contrast, one-way car-sharing offers greater flexibility, but faces spatial-temporal supply-demand imbalances, which can cause vehicles to accumulate at popular destinations, requiring dedicated staff to redistribute them to areas of high demand~\cite{illgen2019literature}.
Without an effective repositioning system, one-way car-sharing systems also face the fleet under-utilization problem. 

Fleet under-utilization is a well-studied issue in both one-way car-sharing~\cite{illgen2019literature} and taxi systems (both autonomous as well as human-driven)~\cite{zardini2022analysis}. 
Various techniques, such as vehicle rebalancing have been proposed to address this challenge. 
For instance, researchers in~\cite{boyaci2015optimization} developed a dynamic repositioning system for one-way car-sharing, while~\cite{brar2022supply} proposed a unified framework integrating vehicle scheduling, rebalancing, and optimized charging to enhance the operational performance of one-way EV car sharing. 
In autonomous fleets, authors in~\cite{pavone2012robotic} devised a supply-demand balancing model to determine the optimal dispatch rate for idle vehicles, ensuring manageable queue lengths for waiting customers. 
Additionally,~\cite{brar2021dynamic} conducted a comprehensive analysis of critical parameters affecting vehicle rebalancing systems, particularly in the face of uncertainties. 
In the context of taxi systems,~\cite{miao2015taxi} introduced an RHC-based approach to ensure fair supply distribution while minimizing idle travel distances, while~\cite{brar2020ensuring} tackled a multi-objective taxi fleet management problem, striving for an optimal balance between company profit and service fairness among individual drivers. 
Furthermore,~\cite{wang2023incentivized} explored four relocation strategies aimed at incentivizing users, revealing notable enhancements in utilization and service rates. 
To mitigate the supply-demand imbalance ~\cite{huang2020vehicle} conducted a comparative analysis between operator-driven and user-driven relocation methods, illustrating that a combined approach employing both strategies could solve the issue.~\cite{wang2021demand} introduced a two-level nested logit model to regulate demand patterns, aiming to alleviate supply-demand imbalances by taking into account alternative pick-up, alternative drop-off, and OD (origin and destination) pair-based incentives. 

While extensive research has targeted at solving the fleet under-utilization problem within one-way car-sharing systems, the issue remains relatively unexplored within the context of round-trip car-sharing systems.
This paper introduces a novel approach aimed at enhancing fleet utilization in round-trip car-sharing systems.
The proposed model allows for one-way trip requests within the traditional round-trip setup. 
By leveraging the proposed user-matching algorithm, subsets of users are paired together to create chains of pick-ups and drop-offs, ensuring that vehicles ultimately return to their originating station.
Furthermore, the model incorporates a unique pricing strategy to encourage user participation in these chains, thereby enhancing fleet efficiency.
This innovative approach offers a promising solution to address fleet under-utilization in round-trip car-sharing systems, paving the way for more efficient and sustainable mobility solutions.

The main contributions of this work are twofold.
Firstly, we formulate an integrated user-matching and pricing model which facilitates one-way car sharing while operating in a round-trip car-sharing setting.
Secondly, we evaluate the effectiveness of the proposed model via extensive simulations in a round-trip car-sharing environment leveraging real-world scenarios.
Specifically, we formulate and test three strategies: i) maximizing service rates, ii) maximizing profits, and iii) maximizing expected profits.
%
%\begin{enumerate}
%    \item To formulate an integrated user-matching and pricing model that facilitates one-way car sharing while operating in a round-trip car-sharing setting. To the best of our knowledge this is the first paper to propose this solution.
%    \item To evaluate the effectiveness of the proposed model via extensive simulations in a round-trip car-sharing environment utilizing real-world data. Specifically, we formulate and test three strategies: i) maximizing service rates, ii) maximizing profits, and iii) maximizing expected profits.
%\end{enumerate}

The paper is organized as follows. In Section II we introduce the N-user matching problem in the round-trip Car Rental setting, followed by a rigorous mathematical model and two-stage integrated user-matching and price control strategies in Section III. 
After a detailed analysis of the performance of the proposed model in Section IV, conclusions are drawn in Section V.

\section{N-User Matching Problem}
This section introduces the proposed round-trip car-sharing operation as well as the key components of the N-User Matching problem allowing one-way trips within the round-trip car-sharing framework.

\subsection{Round-Trip Car-Sharing System}
The proposed system comprises a fleet of~$n$ cars stationed at~$m$ locations, represented by set~$S$. 
One-way trips are allowed and pre-captured by the system. 
Users can rent a car from one station and return it to another station, conditional to the constraint that eventually the car is returned back to it's original station within a stipulated time window. 

Given a pool of one-way trip requests from the users, we aim to identify a set of~$N$ users who, after a chain of pick-ups and drop-off, can bring back the car to the original station. 
The base rental price between any two stations is assumed to be fixed. 
However, adjustments to the price can be made to incentivize the user participation in the matching process. There is a cost involved in moving the vehicles from one station to the other which is assumed to be a fraction of the base rental price. 
We define~$V$ to represent a set of users with one-way trip requirements who are ready to use the one-way car rental service at the base rental price. Such users are termed \emph{active} users as they can be matched without requiring any additional incentives. On the other hand,~$W$ denotes a set of users with one-way trip requirements who are not willing to use the one-way car rental service at the base rental price. These users are labeled \emph{inactive} users, indicating that an incentive is necessary to encourage them to utilize the one-way rental service. It is assumed that the set of inactive users is known a priori. 
$U$ denotes the set of all the users in the system:
 \begin{equation*}
     V \cup W = U, \hspace{1cm}  V \cap W = \emptyset.
 \end{equation*}
Fig. \ref{fig:chain} illustrates a chain of users matched together, where users~$\{v_i, v_{i+1}, v_{i+2}, v_{i+3}\}$ belong to the set of active users~$V$, and users~$\{ w_i, w_{i+1} \}$ belong to the set of inactive users~$W$. The matching process is subject to specific constraints, including factors such as pick-up and drop-off stations, as well as corresponding pick-up and drop-off times. 
If the matching constraints are satisfied, two active users~$\{ v_i, v_{i+1} \}$ can be matched with certainty, denoted by a solid line in the chain. 
Additionally, sufficient incentives must be provided to facilitate matching between an active and inactive user~$\{ v_{i+1}, w_{i} \}$ or between two inactive users~$\{ w_{i}, w_{i+1} \}$ and is denoted by a dotted line. 
Time is discretized into instants~$T_k$ and the planning horizon~$H$ is defined as a set of ordered timeslots~$H = \{T_1, T_2, ... T_{\tau}\}$ where, $\tau$ represents the total number of timeslots within a planning horizon. The duration of each timeslot as well as the planning horizon is pre-determined and fixed for the time being and could be relaxed in the future.
  
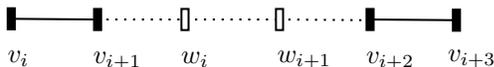
\begin{figure}[b]
    \centering
\tikzset{every picture/.style={line width=0.75pt}} %set default line width to 0.75pt        

\begin{tikzpicture}[x=0.3pt,y=0.4pt,yscale=-1,xscale=1]
%uncomment if require: \path (0,302); %set diagram left start at 0, and has height of 302

%Straight Lines [id:da08119807088837883] 
\draw  [dash pattern={on 0.84pt off 2.51pt}]  (148.14,170.53) -- (377.5,170.53) ;
%Straight Lines [id:da2556734978023625] 
\draw  [dash pattern={on 0.84pt off 2.51pt}]  (377.5,170.53) -- (486.4,171.18) ;
%Shape: Rectangle [id:dp164611873175335] 
\draw  [fill={rgb, 255:red, 255; green, 255; blue, 255 }  ,fill opacity=1 ] (253.83,160.53) -- (262.82,160.53) -- (262.82,180.53) -- (253.83,180.53) -- cycle ;
%Shape: Rectangle [id:dp07629209422160588] 
\draw  [fill={rgb, 255:red, 0; green, 0; blue, 0 }  ,fill opacity=1 ] (143.65,160.53) -- (152.63,160.53) -- (152.63,180.53) -- (143.65,180.53) -- cycle ;
%Shape: Rectangle [id:dp8962757423608949] 
\draw  [fill={rgb, 255:red, 255; green, 255; blue, 255 }  ,fill opacity=1 ] (373.01,160.53) -- (382,160.53) -- (382,180.53) -- (373.01,180.53) -- cycle ;
%Shape: Rectangle [id:dp4019915536576033] 
\draw  [fill={rgb, 255:red, 0; green, 0; blue, 0 }  ,fill opacity=1 ] (487.19,160.53) -- (496.18,160.53) -- (496.18,180.53) -- (487.19,180.53) -- cycle ;
%Shape: Rectangle [id:dp0712339124596657] 
\draw  [fill={rgb, 255:red, 0; green, 0; blue, 0 }  ,fill opacity=1 ] (34.75,159.87) -- (43.74,159.87) -- (43.74,179.87) -- (34.75,179.87) -- cycle ;
%Straight Lines [id:da45947217707743704] 
\draw    (39.25,169.87) -- (148.14,170.53) ;
%Straight Lines [id:da7092289608694666] 
\draw    (491.69,170.53) -- (600.58,171.18) ;
%Shape: Rectangle [id:dp32910270850786283] 
\draw  [fill={rgb, 255:red, 0; green, 0; blue, 0 }  ,fill opacity=1 ] (596.09,161.18) -- (605.07,161.18) -- (605.07,181.18) -- (596.09,181.18) -- cycle ;

% Text Node
\draw (31,197.4) node [anchor=north west][inner sep=0.75pt]    {$v_{i}$};
% Text Node
\draw (139,198.4) node [anchor=north west][inner sep=0.75pt]    {$v_{i+1}$};
% Text Node
\draw (250,198.4) node [anchor=north west][inner sep=0.75pt]    {$w_{i}$};
% Text Node
\draw (371,197.4) node [anchor=north west][inner sep=0.75pt]    {$w_{i+1}$};
% Text Node
\draw (483,198.4) node [anchor=north west][inner sep=0.75pt]    {$v_{i+2}$};
% Text Node
\draw (586,196.4) node [anchor=north west][inner sep=0.75pt]    {$v_{i+3}$};

\end{tikzpicture}
    
    \caption{Graphic representation of a chain of matched users.}
    \label{fig:chain}
\end{figure}

\subsection{Components of N-User Matching Model}
The one-way trip information associated with each user~$u_i \in U$ is denoted as follows: 
\begin{equation*}
    \xi_{u_i} = \{L^s_{u_i}, L^e_{u_i}, T^s_{u_i}, T^e_{u_i}, P_{u_i}, \hat{P}_{u_i}, P^{*}_{u_i}, C_{u_i}\}.
\end{equation*}

\subsubsection{Starting Station} $L^s_{u_i} \in S$ denotes the station from which user~$u_i \in U$ seeks to pick-up the car to start the trip. 
\subsubsection{End Station} $L^e_{u_i} \in S$ denotes the station where user~$u_i \in U$ will drop-off the car at the end of the trip. To avoid any round trips, only one-way trips are considered:
\begin{equation*}
    L^e_{u_i} \neq L^s_{u_i}.
\end{equation*}
\subsubsection{Starting Timeslot} $T^s_{u_i} \in H$ denotes the timeslot in which user~$u_i \in U$ will start the trip.
\subsubsection{End Timeslot} $T^e_{u_i} \in H$ denotes the timeslot in which user~$u_i \in U$ will end the trip. To ensure effective user matching, the timeslot duration is chosen such that the end timeslot is greater than start timeslot:
\begin{equation}\label{eqn:et_vs_st}
    T^e_{u_i} > T^s_{u_i}.
\end{equation}
\subsubsection{Base Rental Price} $P_{u_i} \in \mathbb{R_+}$ denotes the maximum price user~$u_i \in U$ has to pay if he makes the trip from station~$L^s_{u_i} \in S$ to station~$L^e_{u_i} \in S$. 
\subsubsection{Travel Cost} $C_{u_i} \in \mathbb{R}$ denotes the cost the company has to incur if user~$u_i \in U$ makes a trip from station~$L^s_{u_i} \in S$ to station~$L^e_{u_i} \in S$. 
So, the total profit that the company receives by serving customer~$u_i \in U$ is given as the difference between rental price and travel cost:
\begin{equation*}
    P_{u_i} - C_{u_i}.
\end{equation*}
\subsubsection{Price Threshold and Offered Price}: $P^*_{u_i} \in [0,P_{u_i}]$ denotes the price threshold, which is the maximum price at which user~$u_i \in U$ is ready to accept the trip offer. 
Price offered to user~$u_i \in U$ is defined as~$\hat{P}_{u_i} \in [0,P_{u_i}]$. 
The price threshold function as shown in Fig. \ref{fig:threshold} is defined for each user~$u_i \in U$ as follows:
\begin{equation*}
      I(P_{u_i}) = \begin{cases}1~&{\hat{P}_{u_i} \leq P^*_{u_i}}\\0~&{\hat{P}_{u_i} > P^*_{u_i}}\end{cases}.
\end{equation*}
For the active users the price threshold~$P^*_{u_i}$ is the same as the base rental price~$P_{u_i}$ and is thus deterministic. 
However, for the inactive users the threshold price~$P^*_{u_i}$ is less than the rental price~$P_{u_i}$ and the threshold is assumed to be a Gaussian random variable~$P^*_{u_i} \sim \mathcal{N}(\mu_{u_i}, \sigma_{u_i})$ whose parameters are known a priori. 

\begin{figure}[tb]
    \centering
    
\tikzset{every picture/.style={line width=0.75pt}} %set default line width to 0.75pt        

\begin{tikzpicture}[x=0.4pt,y=0.4pt,yscale=-1,xscale=1]
%uncomment if require: \path (0,302); %set diagram left start at 0, and has height of 302


%Straight Lines [id:da6431516182874784] 
\draw    (267.4,217.19) -- (266.41,53.19) ;
\draw [shift={(266.4,51.19)}, rotate = 89.65] [color={rgb, 255:red, 0; green, 0; blue, 0 }  ][line width=0.75]    (10.93,-3.29) .. controls (6.95,-1.4) and (3.31,-0.3) .. (0,0) .. controls (3.31,0.3) and (6.95,1.4) .. (10.93,3.29)   ;
%Straight Lines [id:da3501676308266073] 
\draw    (267.4,217.19) -- (433.4,217.19) ;
\draw [shift={(435.4,217.19)}, rotate = 180] [color={rgb, 255:red, 0; green, 0; blue, 0 }  ][line width=0.75]    (10.93,-3.29) .. controls (6.95,-1.4) and (3.31,-0.3) .. (0,0) .. controls (3.31,0.3) and (6.95,1.4) .. (10.93,3.29)   ;
%Straight Lines [id:da6025202755112831] 
\draw [color={rgb, 255:red, 208; green, 2; blue, 27 }  ,draw opacity=1 ][line width=1.5]    (266.9,147.19) -- (319.4,147.19) ;
%Straight Lines [id:da16127335039167456] 
\draw  [dash pattern={on 0.84pt off 2.51pt}]  (318.4,147.19) -- (318.4,216.19) ;
%Straight Lines [id:da2778106202297357] 
\draw [color={rgb, 255:red, 208; green, 2; blue, 27 }  ,draw opacity=1 ][line width=1.5]    (318.4,217.19) -- (394.9,217.19) ;
%Shape: Circle [id:dp1221573082475429] 
\draw  [color={rgb, 255:red, 208; green, 2; blue, 27 }  ,draw opacity=1 ][fill={rgb, 255:red, 208; green, 2; blue, 27 }  ,fill opacity=1 ] (315,147.79) .. controls (315,145.92) and (316.52,144.4) .. (318.4,144.4) .. controls (320.27,144.4) and (321.79,145.92) .. (321.79,147.79) .. controls (321.79,149.67) and (320.27,151.19) .. (318.4,151.19) .. controls (316.52,151.19) and (315,149.67) .. (315,147.79) -- cycle ;
%Shape: Rectangle [id:dp6547651834645429] 
\draw  [fill={rgb, 255:red, 0; green, 0; blue, 0 }  ,fill opacity=1 ] (319.29,213.69) -- (317.51,213.69) -- (317.51,220.69) -- (319.29,220.69) -- cycle ;
%Shape: Rectangle [id:dp917859274345145] 
\draw  [fill={rgb, 255:red, 0; green, 0; blue, 0 }  ,fill opacity=1 ] (395.79,213.69) -- (394.01,213.69) -- (394.01,220.69) -- (395.79,220.69) -- cycle ;

% Text Node
\draw (435,208.59) node [anchor=north west][inner sep=0.75pt]    {$P$};
% Text Node
\draw (180,53.4) node [anchor=north west][inner sep=0.75pt]    {$I(P_{u_i})$};
\draw (250,135) node [anchor=north west][inner sep=0.75pt]    {$1$};
% Text Node
\draw (384.9,224.59) node [anchor=north west][inner sep=0.75pt]  [font=\scriptsize]  {$P_{u_i}$};
% Text Node
\draw (311.9,222.59) node [anchor=north west][inner sep=0.75pt]  [font=\scriptsize]  {$P^{*}_{u_i }$};

\end{tikzpicture}
    \caption{Price threshold function For inactive users.}
    \label{fig:threshold}
\end{figure}
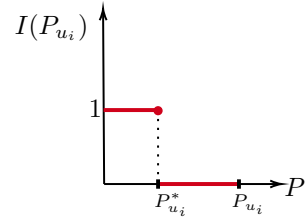

\section{Mathematical Model}
In this section we explain the proposed mathematical model to solve the N-user matching problem using price as an incentive.
We formulate the problem as a constrained optimization problem. First, we define the decision variable of the optimization model. Then, we discuss the logic that is used to match the users. 
Next, we discuss the chain-activation process as well as the cost-factor. Finally, we discuss the constraints and the objective function of the proposed optimization model as well as present the Two-Stage User-Matching approach. 
\subsection{Decision Variable}
We define a chain $c$ of length~$N$ as a set of~$N$ users where each user belongs to the set~$U$:
\begin{equation*}
    c =\{ u_1, u_2, ..., u_N \}, \hspace{0.2cm} \text{where} \hspace{0.2cm} u_1, u_2, ..., u_N\in U.
\end{equation*} 
Decision variable~$x_{u_1u_2...u_N} \in \{0,1\} = 1$ if users in $\{ u_1, u_2, ..., u_N \}$ are matched together. 
The maximum length that a chain of users can have is dependent on the size of the planning horizon~$H = \{T_1, T_2, ... T_{\tau}\}$. 
Since the end timeslot~$T^e_{u_i}$ is greater than the start timeslot~$T^s_{u_i}$ of each user's trip as defined in (\ref{eqn:et_vs_st}), the maximum chain length can be $\tau-1$. 
The length that a chain is allowed to have can depend on various factors. From an operational point of view, long chains can be desirable to ensure effective vehicle utilization by back-to-back bookings. 
However, detecting longer chains in a pool of users can be computationally intensive which may impact the performance of the user-matching system during real-time operation, especially if sequential re-implementation such as rolling horizon-based techniques are used. We define the set of all the possible chain lengths as the depth of the proposed algorithm.
\begin{equation}\label{chain_lengths}
    \mathcal{D}(\tau) = \{ 2, 3, 4, ... \tau - 1 \}.
\end{equation}
The set of decision variables contains all possible chains of users and is defined as:  
\begin{equation*}
    \mathcal{X} = \{  x_{u_1u_2},  x_{u_1u_2u_3}, ... , x_{u_1u_2...u_N}\},
\end{equation*}
where,~$u_1, u_2, ... u_N \in U$ and $N \in  \mathcal{D(\tau)}$ is the depth-cutoff, i.e., chains of length greater than~$N$ are not considered in the user-matching problem. 

\subsection{Feasible Chain}
Any chain~$c = \{ u_1, u_2, ..., u_N \} $ is called \emph{feasible} if the users in the chain can be matched. For a set of users to be matched, the following conditions have to be satisfied:
\begin{enumerate}
    \item Drop-off station of user~$u_i$ is the same as the pick-up station of user~$u_{i + 1}$ for all the users excluding the last user for which the drop-off station must be the same as the pick-up station of the first user to ensure that the vehicle is returned to the starting station where the chain starts:
\begin{equation*}
    \mathcal{L}(\{u_1u_2\ldots u_N\}) = 
    \begin{cases}
        1 & L^e_{u_{i}} = L^s_{u_{i+1}} \text{ for } 1 \leq i < N \\
        & L^e_{u_N} = L^s_{u_{1}} \\
        0 & \text{otherwise}
    \end{cases}.
\end{equation*}
    \item Drop-off timeslot of user~$u_i$ is the same as the pick-up timeslot of user~$u_{i + 1}$ for all the users excluding the last user:
\begin{equation*}
    \mathcal{T}(\{u_1u_2\ldots u_N\}) = 
    \begin{cases}
        1 & T^e_{u_{i}} = T^s_{u_{i+1}} \text{ for } 1 \leq i < N \\
        0 & \text{otherwise}
    \end{cases}.
\end{equation*}
\end{enumerate}
Any chain~$c = \{u_1u_2\ldots u_N\}$ is considered feasible if it follows the logic:
\begin{equation} \label{feasible}
    \mathcal{F}(c) = \mathcal{L}(c) \cdot  \mathcal{T}(c) = 1.
\end{equation}

\subsection{Chain Activation}
Any chain~$c = \{u_1,u_2,\ldots,u_N\}$ is defined to be active if all the users in the chain are active at the offered price. 
This means that each user is willing to take the ride at the offered rental price. A chain may consist of users who are already active at the base rental price as well as users who are inactive at the base rental price but will become active if the price is reduced to be less than a threshold value, i.e., an offered price~$\hat{P}_{u_i} \leq P^*_{u_i}$. 
As defined in the previous section the price threshold $P^*_{u_i}$ is assumed to be normally distributed with known parameters. The probability that the threshold price satisfies~$P^*_{u_i} \leq \hat{P}_{u_i}$ is given by:
\begin{equation*}
    \mathbb{P}(P^*_{u_i} < \hat{P}_{u_i}) =  \int_{0}^{\hat{P}_{u_i}} f_{P^*_{u_i}}(w) \,\text{d}w,
\end{equation*}
where~$f_{P^*_{u_i}}(w)$ is the Probability Density Function (PDF) of the threshold price, and~$(\mu_{u_i}, \sigma_{u_i})$ are its mean and standard deviation. 
The shaded region in Fig. \ref{fig:alpha} is indicative of~$\mathbb{P}(P^*_{u_i} < \hat{P}_{u_i})$ and hence the probability of an inactive user~$u_i \in U$ to remain inactive at the offered price~$\hat{P}_{u_i}$. 
The probability that the user~$u_i$ will become active at the rental price~$P_{u_i}$ is given as follows:
\begin{equation*}
    \mathbb{P}(P^*_{u_i} \geq P_{u_i}) = 1 - \mathbb{P}(P^*_{u_i} < P_{u_i}).
\end{equation*}
It is obvious that if two prices~$\hat{P}_{u_i}, \Tilde{P}_{u_i}$ are offered to a user, such that~$\hat{P}_{u_i} < \Tilde{P}_{u_i}$, then the probability of the user becoming active at the price~$\hat{P}_{u_i}$ is no less than the probability of the user becoming active at the price~$\Tilde{P}_{u_i}$, i.e,
\begin{equation*}
    \mathbb{P}(P^*_{u_i} < \hat{P}_{u_i}) \leq  \mathbb{P}(P^*_{u_i} < \Tilde{P}_{u_i}).
\end{equation*}
For any feasible chain~$c = \{u_1,u_2,\ldots,u_N\}$, let~$\hat{c} \subseteq c$ be a set of inactive users in chain~$c$. 
Given that the price~$\hat{P}_{\hat{u}_i}$ offered to each user~$\hat{u}_i \in \hat{c}$, the probability that the chain~$c$ will be active, i.e., all the users in the chain will be active at the offered price~$\hat{P}_{u_i}$ is given by:
\begin{equation*}
    \mathbb{P}(c) = \prod_{u_i \in c} \mathbb{P}(P^*_{u_i} \geq \hat{P}_{u_i}) = \prod_{\hat{u}_i \in \hat{c}} \mathbb{P}(P^*_{\hat{u}_i} \geq \hat{P}_{\hat{u}_i}).
\end{equation*}

\begin{figure}[tb]
    \centering
    \includegraphics[width=0.2\textwidth]{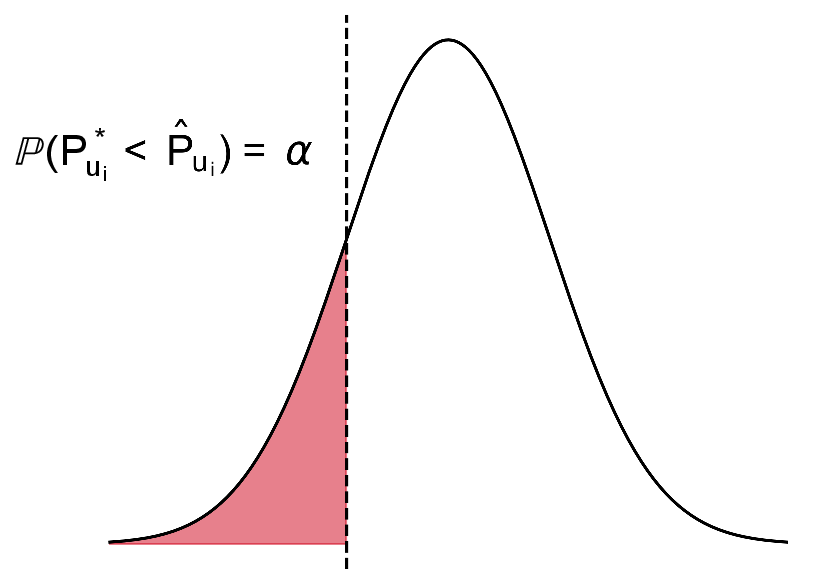} % Adjust the width as needed
    \caption{$f_{P^*_{u_i}}(w)$: PDF of $P^*_{u_i}$ with parameters $(\mu_{u_i}, \sigma_{u_i})$.}
    \label{fig:alpha}
\end{figure}

\begin{theorem}\label{thm:mytheorem1}
If the price offered to each inactive user is~$\hat{P}_{\hat{u}_i} = \mathcal{P}(\theta_{\hat{u}_i}, \alpha)$ such that~$\mathbb{P}(P^*_{\hat{u}_i} < \mathcal{P}(\theta_{\hat{u}_i}, \alpha)) = \alpha$, then the probability that the chain~$c$, becomes active is given by:
\begin{equation}\label{activation}
    \mathbb{P}(c) = (1 - \alpha)^{|\hat{c}|},
\end{equation}
where,~$\hat{c} \subseteq c$ is a set of inactive users in~$c$ 
\end{theorem}

\begin{proof}
Given a chain of users~$c$, and assuming that the event of offering the rental price as well as response of each user is independent. Let~$\hat{c} \subseteq c$ be a set of inactive users, such that each inactive user~$\hat{u}_i \in \hat{c}$ has a price threshold~$P^*_{\hat{u}_i}$ which is normally distributed with parameters~$\theta_{\hat{u}_i} = (\mu_{\hat{u}_i}, \sigma_{\hat{u}_i})$. 
It is always possible to find a price~$\hat{P}_{u_i}$ such that~$\mathbb{P}(P^*_{\hat{u}_i} < \hat{P}_{u_i}) = \alpha$, where~$\alpha = (0,1)$. 
We define this price as the~$\alpha$-percentile~$\mathcal{P}(\theta_{\hat{u}_i}, \alpha)$, which depends on the parameters $\theta_{\hat{u}_i}$ of the PDF of~$P^*_{\hat{u}_i}$. 
We define~$\alpha \in (0,1)$ as the risk factor. 

If we offer~$\alpha$-percentile~$\mathcal{P}(\theta_{\hat{u}_i}, \alpha)$ price to an inactive user $\hat{u}_i$, then the probability that the user $u_i$ will become active at the offered price is:
\begin{equation*}
    \mathbb{P}(P^*_{u_i} \geq \mathcal{P}(\theta_{\hat{u}_i}, \alpha)) = 1-\alpha.
\end{equation*}

If we offer~$\alpha$-percentile of the price threshold distribution to each inactive user~$\hat{u}_i \in \hat{c}$, then using \ref{activation}:
\begin{equation*}
    \mathbb{P}(c) = \prod_{u_i \in \hat{c}} \mathbb{P}(P^*_{u_i} \geq P_{u_i}) = (1 - \alpha)^{|\hat{c}|}.
\end{equation*}
\end{proof}
 
\begin{exmp}
Consider a scenario with~$c = \{u_1, u_2\}$. 
Let~$u_1$ be active with the base rental price and travel cost as~$(P_{u_1}, C_{u_1})$. 
Let~$u_2$ be inactive with the base rental price and travel cost~$(P_{u_2}, C_{u_2})$, offered price~$\hat{P}_{u_2}$, and price threshold parameters ~$\mu_{u_2}, \sigma_{u_2},$ respectively.
The expected profit of the company through this chain completion is as follows:
\begin{equation*}
    \mathbb{P}(P^*_{u_2} \geq \hat{P}_{u_2}) \cdot (P_{u_1} - C_{u_1} + \hat{P}_{u_2} - C_{u_2}).
\end{equation*}
Since,~$P_{u_1} - C_{u_1} - C_{u_2}$ is fixed, let~$P_{u_1} - C_{u_1} - C_{u_2} = \Delta$. 
Then the above equation can be re-written as:
\begin{equation*}
    \mathbb{P}(P^*_{u_2} \geq \hat{P}_{u_2})\cdot (\Delta + \hat{P}_{u_2} ).
\end{equation*}
Offering high $\hat{P}_{u_2}$ leads to higher profit but there is a higher risk (less likelihood) of chain completion. So, there is a risk-return trade-off and the user-matching model needs to consider this aspect while calculating the optimal chains. Eq. (\ref{activation}) provides an efficient way to quantify and control the risk using the~$\alpha \in (0,1)$ value which we define as the risk-factor. 
In the above example if~$\alpha = 0.5$ then the risk of the chain completion is 50\% and if the $\alpha$-percentile price is offered to user~$u_2$, then~$\hat{P}_{u_2}=\mu_{u_2}$ and the expected profit from this chain will be~$0.5 \times (\Delta + \mu_{u_2})$.
\end{exmp}

\subsection{Cost Factor}
For each user~$u_i \in U$, a cost factor~$cf(u_i)$ is defined as a ratio of the travelling cost~$C_{u_i}$ to the base rental price $P_{u_i}$.
\begin{equation*}
    cf(u_i) = \frac{C_{u_i}}{P_{u_i}} \in (0,1) \hspace{0.1cm} \forall \hspace{0.1cm}  u_i \in U.
\end{equation*}
The company profit for serving an active user~$u_i$ is the difference between the base rental price and traveling cost:
\begin{equation*}
P_{u_i} - C_{u_i} = P_{u_i} - cf(u_i) \cdot P_{u_i} = (1-cf(u_i)) \cdot P_{u_i} . 
\end{equation*}
If an inactive user~$\hat{u}_i$ is offered a price~$\hat{P}_{\hat{u}_i} \in [0,  P_{\hat{u}_i})$ and the user become active at this price then the profit associated with serving this user is:
\begin{equation*}
\hat{P}_{\hat{u}_i} - C_{u_i} = \hat{P}_{\hat{u}_i} - cf(\hat{u}_i) \cdot P_{\hat{u}_i}.
\end{equation*}
The profit associated with a chain~$c = \{u_1,u_2, \ldots ,u_N\}$ consisting of both inactive users~$\hat{c} \subseteq c$ and active users~$c - \hat{c} \subseteq c$, is given as the sum of the profits associated with serving each user in the chain:
\begin{equation}\label{chain_profit}
    P(c) = \sum_{u_i \in \hat{c}} \hat{P}_{u_i} - cf(u_i) \cdot P_{u_i}   + \sum_{u_i \in c-\hat{c}} (1-cf(u_i)) \cdot P_{u_i}
\end{equation}
It is undesirable and also highly unlikely for~$cf(u_i) > 1$, so the~$cf(u_i)$ is constrained to be in~$(0,1]$.
Note that for trips of inactive users, the overall profit may be negative (i.e., a loss) if~$\hat{P}_{\hat{u}_i} < cf(\hat{u}_i) \cdot P_{\hat{u}_i}$ 

\subsection{Constraints}
In this section, we explain the constraints of the N-user matching problem. Given a set of users~$U$:
\begin{enumerate}
    \item For a chain of length~$N \in [2, \tau-1]$, the set of all possible chains of length~$N$ that contain user~$u_i$:
    \begin{enumerate}
        \item Two-User Chains: $\forall u_1, u_2 \in U$:
        \begin{equation*}
        x_2(u_i) = \{x_{u_iu_2}\} \cup \{x_{u_1u_i}\}.   
        \end{equation*}
        \item Three-User Chains: $\forall u_1, u_2, u_3 \in U$:
        \begin{equation*}
        x_3(u_i) = \{x_{u_iu_2u_3}\} \cup \{x_{u_1u_iu_3}\} \cup \{x_{u_1u_2u_i} \}.
        \end{equation*}
        \begin{equation*}
        \vdots
        \end{equation*}
        \item N-User Chains: $\forall u_1, u_2,  \ldots  u_N \in U$:
        \begin{equation*}
        x_N(u_i) = \{x_{u_i \ldots u_N}\} \cup \{x_{u_1u_i  \ldots  u_N}\} \ldots \cup \ldots\{x_{u_1 \ldots u_i} \}.   
        \end{equation*}
    \end{enumerate} 
    \item Each user~$u_i \in U$ can only be part of one chain $c$, which is enforced by the following constraint:
    \begin{equation}\label{c1}
    \begin{aligned}
        & \sum_{u_1, u_2 \in U} x_2(u_i) + \ldots + \sum_{u_1, u_2, \ldots u_N \in U} x_N(u_i) \leq 1.  
    \end{aligned}
    \end{equation}
    \item A chain of users~$\{u_1, \ldots ,u_N\}$ can only be matched if the chain is feasible, i.e., the chain satisfies the feasibility conditions as specified by (\ref{feasible}):
    \begin{equation}\label{c2}
        x_{u_1 \ldots u_N} \leq \mathcal{F}(\{u_1, \ldots ,u_N\}),
    \end{equation}
    where,
    \begin{equation*}
        \mathcal{F}(\{u_1, \ldots ,u_N\}) = \mathcal{L}(\{u_1, \ldots ,u_N\}) \cdot \mathcal{T}(\{u_1, \ldots ,u_N\})
    \end{equation*}
    \item Each possible chain has a binary variable associated with it which will decide if the chain of users will be matched or not by the system.
    \begin{equation}\label{c3}
        x_{u_1u_2 \ldots u_N} \in \{0, 1\} \hspace{0.3cm} u_1,u_2, \ldots ,u_N \in U.
    \end{equation}
\end{enumerate}

\subsection{Two-Stage Approach}
Consider the problem of maximizing the company profit by offering an incentivized price to the inactive users and matching the optimal set of users. 
Consider a chain~$c = \{u_1, \ldots, u_d \}$ of~$d$ users where,~$d \in [2,N]$, and the profit associated with the chain is~$P(\{u_1, \ldots, u_d \})$. 
The problem is to select an optimal set of chains that maximize the company profit while satisfying constraints~$\mathcal{C} = \{(\ref{c1}), (\ref{c2}), (\ref{c3})\}$:
\begin{equation} \label{intractable}
\begin{aligned}
    & \max_{x, \hat{P}}  \sum_{d=2}^N \sum_{u_1, ..., u_d \in U} x_{u_1 \ldots u_d} \cdot P(\{u_1, \ldots, u_d \}) \\  
    \textrm{s.t.}  \quad  & \sum_{d=2}^N \sum_{u_1, u_2, \ldots u_d \in U} x_d(u_i) \leq 1  \hspace{0.2cm} \forall u_i \in U\\ 
    & x_{u_1 \ldots u_d} \leq \mathcal{F}(\{u_1, \ldots ,u_d\})  \forall u_1, \ldots ,u_d \in U, d \in [2,N]\\
    & \hat{P}_{u_i} \in [0, P_{u_i}) \hspace{0.2cm} \forall u_i \in U \\
    & x_{u_1u_2 \ldots u_d} \in \{0, 1\} \hspace{0.2cm} \forall u_1,u_2, \ldots ,u_d \in U, d \in [2,N].
\end{aligned}
\end{equation}
The above problem is a Mixed Integer Linear Programming (MILP) problem because the objective function involves a product of binary variable~$x_{u_1 \ldots u_N}$ and~$P(\{u_1, \ldots, u_N \})$ is a function of~$\hat{P}_{u_i}$ as in (\ref{chain_profit}). 

To make a tractable formulation, we use the special feature of this problem. We first show how the above problem can be converted to an ILP in a setting where price thresholds are deterministic and known a priori.

\begin{theorem}\label{thm:mytheorem2}
Given a chain of users~$\{u_1, ... u_i, ... u_d \}$ and their deterministic threshold prices~$P^*_{u_i}$, the optimal price $\Tilde{P}_{u_i}$ offered to user $u_i$ is the same as the threshold price:
\begin{equation*}
    \Tilde{P}_{u_i} = P^*_{u_i} \hspace{0.2cm} \forall \hspace{0.2cm}  i \in [1,d] .
\end{equation*}
\end{theorem}

\begin{proof}
We use proof by contradiction to prove this statement. 
The objective that needs to be maximized is:
\begin{equation}\label{profit_maximization_intractable}
\begin{aligned}
    &\max_{x, \hat{P}}  \sum_{d=2}^N \sum_{u_1, ..., u_d \in U} x_{u_1 \ldots u_d} \cdot P(\{u_1, \ldots, u_d \}).
\end{aligned}
\end{equation}
Let~$\Tilde{X}$ be the set of chains that maximize the profit, and let~$\Tilde{P}$ be the optimal prices offered. 
Let~$\Tilde{x} = \{u_1, ... u_i, ... u_d \} \in \Tilde{X}$ containing one inactive user $u_i$ be one of the chains which is selected, i.e., ~$x_{u_1...u_i...u_d}=1$. 
Let the price offered to~$u_i$ be~$\Tilde{P}_{u_i}$ with his threshold price being~$P^*_{u_i}$. Let us assume that~$\tilde{P}_{u_i} = P^*_{u_i}$ is the sub-optimal price, then there are two possible cases for the optimal price offered to user~$u_i$:
\begin{enumerate} 
    \item Case1: $\Tilde{P}_{u_i} > P^*_{u_i}$ \\
    In this case $u_i$ will become inactive so the chain~$\Tilde{x} = \{u_1, ... u_i, ... u_N \}$ will be inactive, leading to a 0 profit. 
    \item Case2: $\Tilde{P}_{u_i} < P^*_{u_i}$ \\
    In this case $u_i$ will become active so the chain $\Tilde{x} = \{u_1, ... u_i, ... u_N \}$ will become active, thus using (\ref{chain_profit}), the profit corresponding to this chain is:
    \begin{equation*}
        \sum_{j=1}^N ( \Tilde{P}_{u_j} - C_{u_j} ) < \sum_{j=1, j \neq i}^N ( \Tilde{P}_{u_j} - C_{u_j} ) + P^*_{u_i} - C_{u_i}.
    \end{equation*}
\end{enumerate}
In both cases, the profit is higher when the offered price is the same as the threshold price which is a contradiction because we assumed that~$\tilde{P}_{u_i} = P^*_{u_i}$ is a sub-optimal price.
So, the maximum profit occurs when the offered price is the same as the threshold price. The same logic can be extended to the chains with more than one inactive user. So, to maximize the profits, all inactive users in a chain must be offered the price threshold. QED
\end{proof} 

Next, we extend this logic to a stochastic setting, where the price threshold is stochastic and follows a normal distribution whose parameters are known a priori. 
We show that problem (\ref{intractable}) can be converted into an ILP for a given risk factor~$\alpha \in (0,1)$.
\begin{theorem}\label{thm:mytheorem3}
    Consider a chain of users $\{u_1, ... u_i, ... u_d \}$ and their price threshold parameters $\theta_{u_i}$. For a fixed upper limit of risk factor $\bar{\alpha} \in (0,1)$, the optimal price $\Tilde{P}_{u_i}$ offered to inactive user $u_i$ is $\mathcal{P}(\theta_{u_i}, \bar{\alpha})$
\begin{equation*}
    \Tilde{P}_{u_i} = \mathcal{P}(\theta_{u_i}, \bar{\alpha}) \hspace{0.1cm} \forall \hspace{0.1cm}  i \in [1,d].
\end{equation*}
Where, $P^*_{u_i} \sim N(\mu_{u_i}, \sigma_{u_i})$ and $\mathbb{P}(P^{*}_{u_i} \leq \mathcal{P}(\theta_{u_i}, \bar{\alpha})) = \bar{\alpha}$ 
\end{theorem}

\begin{proof}
We use proof by contradiction to prove this statement. The objective that needs to be maximized is (\ref{profit_maximization_intractable}). Let the risk factor associated with offering $\Tilde{P}_{u_i}$  price to inactive user $u_i$ is $\Tilde{\alpha}$. 
    
We first assume that the profit earned by formation of chain $\Tilde{x}$ is not optimal when $\Tilde{P}_{u_i} = \mathcal{P}(\theta_{u_i}, \bar{\alpha})$, then the optimal $\Tilde{P}_{u_i}$ price may be in either of these two cases: 

\begin{enumerate} 
    \item Case 1: $\Tilde{P}_{u_i} > \mathcal{P}(\theta_{u_i}, \bar{\alpha})$ \\
     In this case, the probability that the price threshold is more than the offered price is less than the probability that the price threshold is more than the $\bar{\alpha}$-percentile.
    \begin{equation*}
    \begin{aligned}
    & \mathbb{P}(\Tilde{P}_{u_i} < P^*_{u_i}) < \mathbb{P}(\mathcal{P}(\theta_{u_i}, \bar{\alpha})< P^*_{u_i}) \\
    & \implies 1 - \hat{\alpha} < 1 - \bar{\alpha} \\
    & \implies \bar{\alpha} < \hat{\alpha}.
    \end{aligned}
    \end{equation*}
    But $\bar{\alpha}$ is the upper limit of risk factor. Therefore, it is a contradiction.
    \item Case 2: $\Tilde{P}_{u  _i} < \mathcal{P}(\theta_{u_i}, \bar{\alpha})$ \\
    In this case, the risk factor is within the upper limit, and the overall profit from the chain $\Tilde{x} = \{u_1, ... u_i, ... u_N \}$ is given as 
    \begin{equation*}
    \begin{aligned}
        &\sum_{j=1}^N ( \Tilde{P}_{u_j} - C_{u_j} ) = \sum_{j=1, j \neq i}^N ( \Tilde{P}_{u_j} - C_{u_j} ) + \Tilde{P}_{u_i} - C_{u_i} \\
        & < \sum_{j=1, j \neq i}^N ( \Tilde{P}_{u_j} - C_{u_j} ) + \mathcal{P}(\theta_{u_i}, \bar{\alpha}) - C_{u_i}.
    \end{aligned}
    \end{equation*}
    This means that the profit earned from the $\Tilde{x}$ chain formation is lower than that when the offered price $\Tilde{P}_{u_i} < \mathcal{P}(\theta_{u_i}, \bar{\alpha})$.
\end{enumerate}
This is a contradiction since the profit for the offered price $\mathcal{P}(\theta_{u_i}, \bar{\alpha})$ is greater in both cases, so $\Tilde{P}_{u_i} = \mathcal{P}(\theta_{u_i}, \bar{\alpha})$ is the optimal price that should be offered to the user $u_i$ to ensure the maximum profits while maintaining the risk levels within the upper limit of $\bar{\alpha}$. QED 
\end{proof}

\subsection{Objective Function}
Given the risk factor $\alpha$, using Theorem \ref{thm:mytheorem3} the optimal price for inactive users is set to~$\mathcal{P}(\theta_{u_i}, \alpha)$. We introduce three different objective functions, which capture three different service considerations. Thus, leading to three different models. 

\paragraph*{Service Rate Maximization}
In this model, we aim to maximize the service rate, which is the total number of users that can be matched in a chain and hence be served. This objective doesn't take offered prices into account and the Service Rate Maximization problem is defined as follows:
\begin{equation*}
    \begin{aligned}
    \max_{\mathcal{X}} \quad & \sum_{u_i \in U} \sum_{j=2}^{N} x_j(u_i) \\
    \textrm{s.t.} \quad & (6),(7),(8)
    \end{aligned}
\end{equation*}
\paragraph*{Profit Maximization}
In this model, we aim to maximize the overall profit that the company will make out of matching the users. This objective is similar to the objective of the problem (\ref{intractable}), However, the difference is that here offered price is pre-calculated and is no-longer a decision variable. This Profit Maximization problem is defined as follows:
    \begin{equation*}
    \begin{aligned}
    \max_{x}  \sum_{d=2}^N \quad & \sum_{u_1, ..., u_d \in U} x_{u_1 \ldots u_d} \cdot P(\{u_1, \ldots, u_d \}) \\
    \textrm{s.t.} \quad & (6),(7),(8)
    \end{aligned}
    \end{equation*}
\paragraph*{Expected Profit Maximization}
In this model, we aim to maximize the expected profit that the company will make out of matching the users. Due to the presence of the inactive users, there is always a risk that a chain may or may not be activated. 
The profit~$P(\{u_1, u_2,..u_N \})$, associated with a chain~$c=\{u_1, u_2,..u_N \}$ can be availed only if the chain~$c$ is active. 
Hence, the expected profit associated with any chain is given by:
    \begin{equation*}
        \EX_{\theta} [P(c)] = \mathbb{P}(c) \cdot P(c),
    \end{equation*}
where~$\mathbb{P}(\{u_1, u_2,..u_N \}) = (1-\alpha)^{|\hat{c}|}$, is the probability that chain $c = \{u_1, u_2,..u_N \}$ is active, and is calculated using Theorem \ref{thm:mytheorem1}, where $\hat{c} \subseteq C$, is a set of inactive users in chain $c$ and $\alpha \in (0,1)$ is the risk factor.
Leveraging on Theorem \ref{thm:mytheorem3}, the optimal profit earned from a chain can be derived as follows. The optimal offered price for the active users is the same as the rental price and for the inactive users, is the same as the $\alpha$-percentile of PDF of the price threshold.
\begin{equation*}
P(\{u_1, u_2,..u_N \}) = \sum_{j=1}^N (\tilde{P}_{u_j} - C_{u_j}),
\end{equation*}
where 
\begin{equation*}
\tilde{P}_{u_j} = 
\begin{cases}
        P_{u_j} & u_j \in c - \hat{c}\\
        \mathcal{P}(\theta_{u_j}, \alpha) & u_j \in \hat{c}. \\
        \end{cases}
\end{equation*}

The Expected Profit Maximization problem is defined as follows:
    \begin{equation*}
    \begin{aligned}
    \max_{x \in \mathcal{X}} \quad & \sum_{d=2}^{N} \sum_{u_1, ..., u_d \in U} x_{u_1 \ldots u_d} \cdot \EX_{\theta} [P(\{u_1, \ldots, u_d \})] \\
    \textrm{s.t.} \quad & (6),(7),(8)
    \end{aligned}
    \end{equation*}
In this study, the Expected Profit Maximization is defined as the proposed model.   
% \subsection{N-User Matching Model}
% The proposed model to solve the N-user matching problem is to maximize the expected profit objective subject to constraints (\ref{c1}), (\ref{c2}), and (\ref{c3}).

% \begin{equation*}
% \begin{aligned}
% \max_{x \in \mathcal{X}} \quad & \sum_{d=2}^{N} \sum_{u_1, ..., u_d \in U} x_{u_1 \ldots u_d} \cdot \EX_{\theta} [P(\{u_1, \ldots, u_d \})] \\
% \textrm{s.t.} \quad & \sum_{d=2}^{N} \sum_{u_1, \ldots u_d \in U} x_d(u_i) \leq 1  \hspace{0.1cm}  \forall u_i \in U \\
%   & x_{u_1 \ldots u_d} \leq \mathcal{F}(\{u_1, \ldots, u_d \}) \hspace{0.1cm} \forall  u_1, ..., u_d \in U, d \in [2,N]\\
%   & x_{u_1 \ldots u_d} \in \{0, 1\} \hspace{0.1cm} \forall  u_1, ..., u_d \in U, d \in [2,N].
% \end{aligned}
% \end{equation*}

\section{Discussion of Experimental Results}
In this section, we report the details of an experimental setup and discuss the observations from the experiments. 
Real-world data from New York City, USA, was used to extract the trip historic information~\cite{donovan2014new}. 
Timeslot duration was fixed to 10 mins, and the planning horizon~$H$ was fixed to 60 mins. 
So, the planning horizon was divided into $\tau = 6$ timeslots:
\begin{equation*}
    H = \{T_1, T_2, T_3, T_4, T_5, T_6 \}
\end{equation*}
All the trips that were shorter than 10 mins were removed to ensure (\ref{eqn:et_vs_st}) is satisfied. The maximum possible length of a chain was~$\tau - 1 = 5$, so the chains of length $N \in \{2,3,4,5\}$ were allowed. 
The operating zones were split into smaller regions and the trips were randomly allocated to a region within a zone. 
Each region is assumed to have one station, and was assigned with a unique Station ID. 
All the trips started and ended at these stations. Trips that started and ended at the same station were removed from the simulation to avoid A-to-A loops in the chain. 
Each trip was allocated to a unique user ID, which represented the user request. A user was randomly classified as an active or an inactive user. 
In total, 80\% of the users were considered as active at the base-price and 20\% of the users were considered as inactive at the base-price. 
This choice was arbitrary, and the objective was to test the proposed algorithm's chain detection capability. 
2413 users were considered in the simulation, out of which 110 had a round-trip type demand, this demand was assumed to be served using the traditional round-trip car-sharing system. The remaining 2303 users were considered for user-matching process. The base price was taken directly from the fare amount mentioned in the dataset. Each traveling cost was a variable depending upon the value of the cost factor. 
The mean of the price threshold distribution for the inactive users was selected by randomly sampling from a uniform distribution over the range~$[0, \text{trip fare amount}]$. 
The standard deviation of the price threshold distribution was assumed to be fixed for all the inactive users. We performed several experiments to test the performance of the proposed N-user matching model. First, we discuss the performance of the proposed model w.r.t, the baseline models in-terms of the total number of users that were matched. Next, we test the robustness of the three models against the risk factor $\alpha$ and cost factor. Finally, we illustrate the contribution of chains of different lengths to the company profit.

\subsection{Served Users}
As illustrated in Fig. \ref{fig:served_user_count}, our proposed model successfully served 965 users, constituting approximately 39\% of the total user base. A baseline is established by the A-to-A demand of 110 users, roughly 4\% of the total. Consequently, the proposed model can accommodate an additional 35\% of demands.
\begin{figure}[tb]
    \centering
    \includegraphics[width=0.35\textwidth]{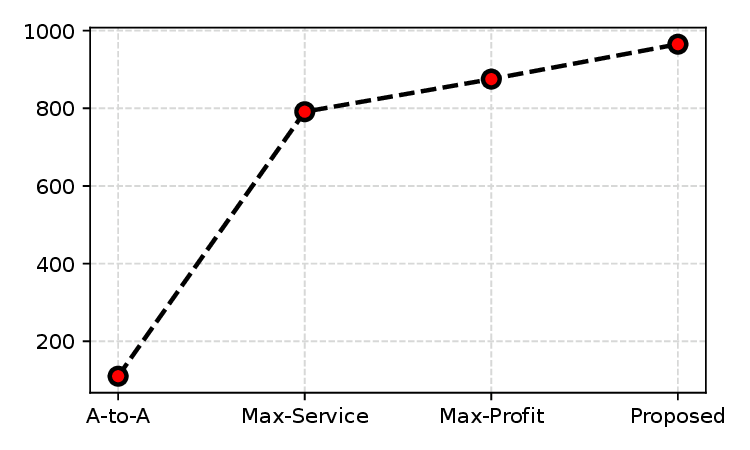}
    \caption{Served user count vs. model.}
    \label{fig:served_user_count}
\end{figure}

\subsection{Risk Factor Analysis}
As depicted in Fig. \ref{fig:profit_vs_alpha} and \ref{fig:sd_vs_alpha}, an increase in the risk factor correlates with a downward trend in profit and service rate for both the service maximization and profit maximization models. Nevertheless, the proposed model exhibits relative robustness to varying risk factors. Compared to the Profit Maximization model, it yields a 19\% higher profit, and compared to the Service Maximization Model, it achieves a 52\% higher profit.
\begin{figure}[tb]
    \centering
\includegraphics[width=0.35\textwidth]{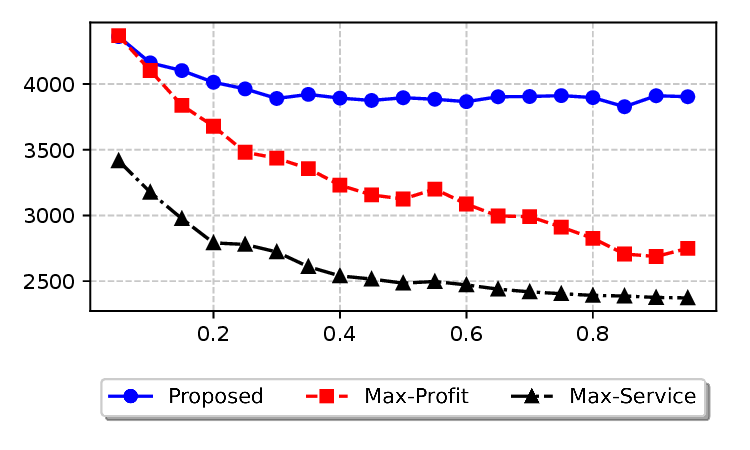}
    \caption{Profit (USD/h) vs risk factor~$\alpha$.}
    \label{fig:profit_vs_alpha}
\end{figure}

Service rate is the percentage of the users that got served. The service rate for the proposed model is 16\% more than the Profit Maximization model, and 17\% more than the Service Maximization model. 
\begin{figure}[tb]
    \centering
\includegraphics[width=0.35\textwidth]{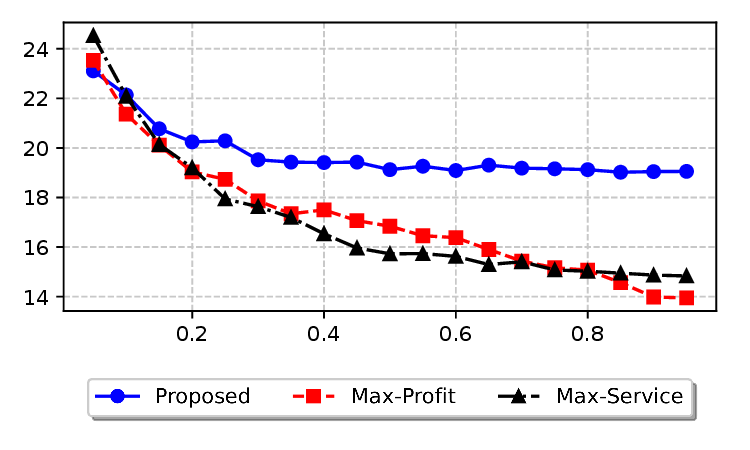}
    \caption{Service rate (\%) vs risk factor~$\alpha$.}
    \label{fig:sd_vs_alpha}
\end{figure}
Since the profit and service rate for the proposed model remains almost constant, it is advantageous because the company doesn't need to give very high incentives (low rental price) by keeping low $\alpha$ values to ensure chain formation.   

\subsection{Cost Factor Analysis}
As the cost factor increases, the expense of each trip rises accordingly. Consequently, as shown in Figure \ref{fig:profit_vs_cf}, profits decrease sharply with rising cost factors, reaching zero at a cost factor of 1. Beyond this point, the company incurs losses greater than its earnings. Therefore, if the cost of chain formation exceeds the total revenue at the base rental price, user matching becomes economically unfeasible. With a constant risk factor of 0.5 and a cost factor ranging from 0 to 1, the proposed model achieves a 21\% higher profit than the profit maximization model and a 57\% higher profit than the service rate maximization model.
\begin{figure}[tb]
    \centering
\includegraphics[width=0.35\textwidth]{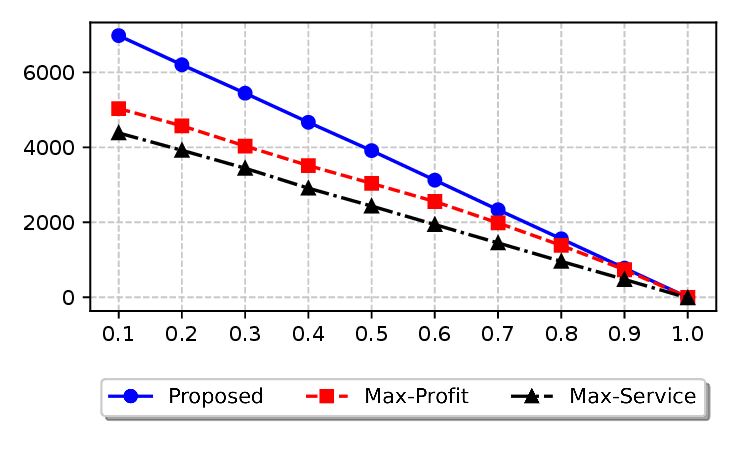}
    \caption{Profit (USD/h) vs cost factor.}
    \label{fig:profit_vs_cf}
\end{figure}

As depicted in Figure \ref{fig:sr_vs_cf}, the service rate for both the proposed model and the Service Maximization model remains relatively stable, while it shows an upward trend for the Profit Maximization model. This increase is attributed to the fact that since the cost factor is increased, more users are served to ensure higher profits. Additionally, it is important to note that both the proposed model and the Profit Maximization model experience a drop to zero service rate when the cost factor reaches 1, as the profit component of the objective function becomes null at this point.
\begin{figure}[tb]
    \centering
\includegraphics[width=0.35\textwidth]{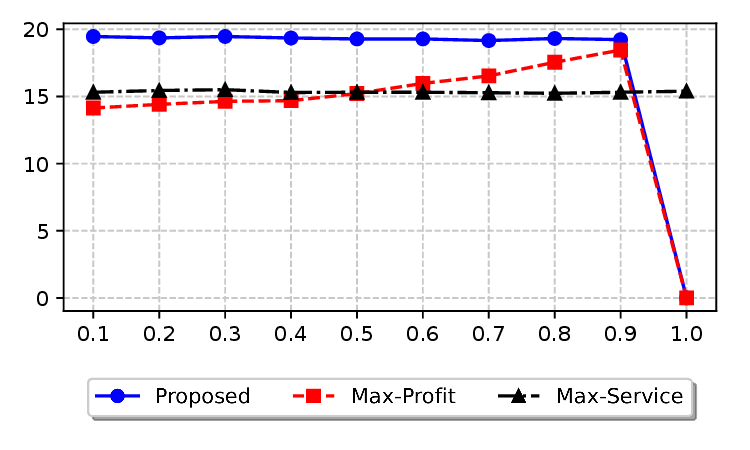}
    \caption{Service rate (\%) vs cost factor.}
    \label{fig:sr_vs_cf}
\end{figure}

\subsection{Chain Length Analysis}
In this experiment, chain lengths vary between 2 and 5. Figure \ref{fig:pr_vs_cf} illustrates the percentage contribution of each chain length to the company's total profit, with chains of length 3 contributing the most. Similarly, Figure \ref{fig:sd_vs_cf} displays a corresponding trend in the service rate.
\begin{figure}[b]
\centering    \includegraphics[width=0.35\textwidth]{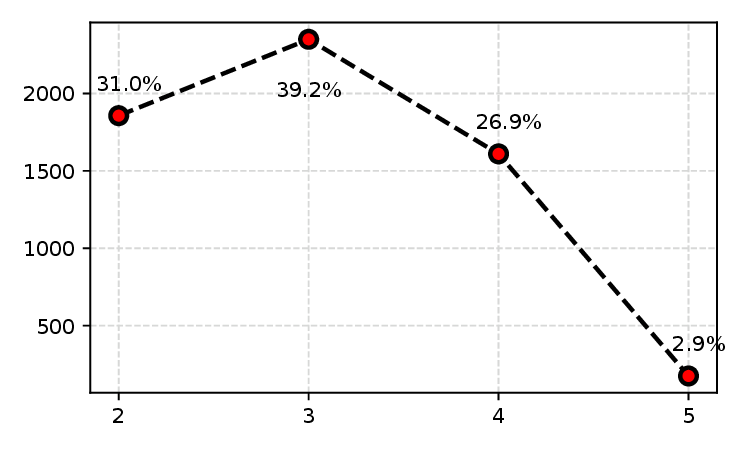}
    \caption{Profit contribution vs chain length.}
    \label{fig:pr_vs_cf}
\end{figure}

\begin{figure}[tb]
    \centering
\includegraphics[width=0.35\textwidth]{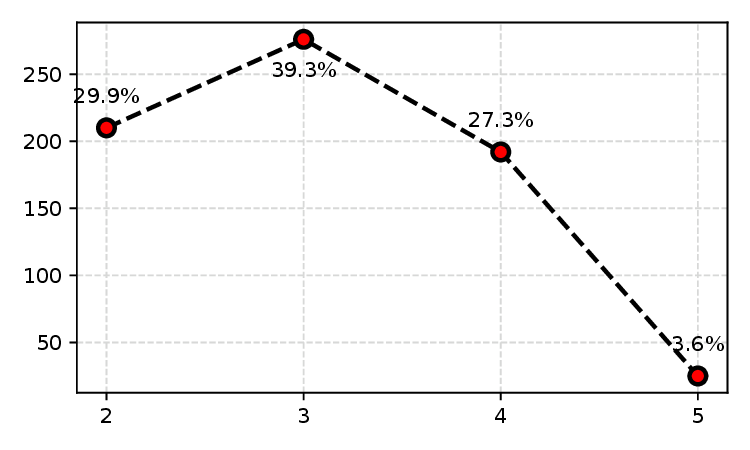}
    \caption{Served users vs chain length.}
    \label{fig:sd_vs_cf}
\end{figure}

 Too short chains lead to a lower service rate and profit whereas too long chains are less likely to exist. Thus, the chain length that has both a high likelihood of existence as well as a high associated reward is somewhere in the middle, which is the chain of length 3 in this case.

\begin{table*}
\let\mc\multicolumn
\caption{Cost Factor vs Risk Factor Analysis}
\label{my-label}
% \begin{tabularx}{\textwidth}{@{} l *{10}{C} c @{}}
% Redefine column type for model column to have more width
\newcolumntype{M}[1]{>{\raggedright\arraybackslash}p{#1}}
\newcolumntype{C}{>{\centering\arraybackslash}X} % Centering text in X columns
\begin{tabularx} {\textwidth}{@{} l M{1.5cm} *{10}{C} @{}}
\toprule  
&Risk Factor & \mc2c{0.2} & \mc2c{0.4} & \mc2c{0.6} & \mc2c{0.8} & \mc2c{1.0} \\ 
\cmidrule(r){1-12} 
  Cost Factor& Model& Profit & Service & Profit & Service& Profit & Service & Profit & Service & Profit & Service \\
\midrule
0.2        &Max-Service  & 4705.82          & 18.99\%   & 4164.73  & 16.44\%    & 3940.42   & 15.47\%    & 3828.86     & 14.97\%   & 3786.50    & 14.80\%  \\ 
           &Max-Profit  & 5474.95          & 17.30\%   & 4973.34  & 15.73\%    & 4552.21   & 14.50\%    & 4484.62     & 14.10\%   & 4161.95    & 13.24\% \\
           &Proposed & \textbf{6342.57} & \textbf{19.87\%}   & \textbf{6132.59}  & \textbf{19.03\%}    & \textbf{6197.99}   & \textbf{19.44\%}    & \textbf{6236.54}    & \textbf{19.11\%}   & \textbf{6245.17}    & \textbf{19.04\%} \\ \hline
0.4        & Max-Service        & 3475.87        & 18.92\%    & 3083.45   & 16.27\%    & 2965.69   & 15.53\%    & 2892.02     & 15.08\%   & 2845.39    & 14.83\%  \\ 
            &Max-Profit & 4404.19          & 18.79\%   & 3752.31  & 16.20\%    & 3473.36   & 14.75\%    & 3278.32     & 13.96\%   & 3274.83    & 13.77\% \\
            &Proposed & \textbf{4756.35}         & \textbf{20.08\%}   & \textbf{4674.70}  & \textbf{19.37\%}    & \textbf{4659.35}   & \textbf{19.24\%}    & \textbf{4681.11}     & \textbf{19.13\%}  & \textbf{4686.79}    & \textbf{19.12\%} \\ \hline
0.6       & Max-Service       & 2261.16          & 19.20\%  & 2063.43  & 16.53\%   & 1967.64  & 15.62\%   & 1916.71    & 15.02\%    & 1896.48     & 14.83\%  \\
        &Max-Profit & 2925.23          & 19.04\%   & 2766.37  & 17.5\%    & 2590.77   & 16.38\%    & 2394.58     & 15.07\%   & 2215.73    & 13.95\% \\
        &Proposed & \textbf{3160.30}         & \textbf{20.08\%}  & \textbf{3106.56}  & \textbf{19.37\%}    & \textbf{3084.65}   & \textbf{19.24\%}    & \textbf{3118.75}     & \textbf{19.13\%}   & \textbf{3127.89}    & \textbf{19.12\%} \\ \hline
0.8       & Max-Service       & 1036.44          & 18.98\%  & 993.17  & 16.36\%   & 975.37  & 15.50\%   & 954.77    & 15.05\%    & 946.28     & 14.80\%  \\
        &Max-Profit & 1544.27          & 19.97\%   & 1457.24  & 18.68\%    & 1394.79   & 17.64\%    & 1321.47     & 16.53\%   & 1133.30    & 14.43\% \\
        &Proposed & \textbf{1582.46}         & \textbf{20.36\%}   & \textbf{1554.10}  & \textbf{19.37\%}    & \textbf{1558.35}   & \textbf{19.31\%}    & \textbf{1562.54}     & \textbf{19.22\%}   & \textbf{1561.65}    & \textbf{19.11\%} \\ \hline
1.0       & Max-Service       & -189.28          & \textbf{19.19\%}  & -55.61  & \textbf{16.25\%}   & -20.91  & \textbf{15.55\%}   & -6.16    & \textbf{15.03\%}    & -1.90     & \textbf{14.83\%} \\
        &Max-Profit & 0.41          & 0.02\%   & 1.06  & 0.03\%    & 0.34   & 0.01\%    & \textbf{0.50}     & 0.06\%   & 0.12    & 0.02\% \\
        &Proposed & \textbf{3.58}         & 0.17\%   & \textbf{1.46}  & 0.04\%    & \textbf{1.09}   & 0.03\%   & 0.41     & 0.03\%   & \textbf{0.40} & 0.06\% \\ 
\bottomrule
\end{tabularx}
\end{table*}

\section{CONCLUSIONS}
In this paper, we have formulated an integrated user-matching and pricing model to efficiently handle one-way car-sharing trip requests while operating in the round-trip car-sharing setting. The robustness of the proposed model was tested under varied risk factors, cost factors, and chain lengths. Extensive simulations on real-world data showed the benefit of using the proposed model in increasing the fleet-utilization as well as company profits of round-trip car sharing. For our forthcoming research, our objective is to integrate the strategic behaviors of individual users within the game-theoretic framework into the N-User Matching problem.

\addtolength{\textheight}{-12cm}   % This command serves to balance the column lengths
                                  % This command does not take effect til the next page

%%%%%%%%%%%%%%%%%%%%%%%%%%%%%%%%%%%%%%%%%%%%%%%%%%%%%%%%%%%%%%%%%%%%%%%%%%%%%%%%

%%%%%%%%%%%%%%%%%%%%%%%%%%%%%%%%%%%%%%%%%%%%%%%%%%%%%%%%%%%%%%%%%%%%%%%%%%%%%%%%



% \section*{ACKNOWLEDGMENT}

% The preferred spelling of the word ÒacknowledgmentÓ in America is without an ÒeÓ after the ÒgÓ. Avoid the stilted expression, ÒOne of us (R. B. G.) thanks . . .Ó  Instead, try ÒR. B. G. thanksÓ. Put sponsor acknowledgments in the unnumbered footnote on the first page.

% %%%%%%%%%%%%%%%%%%%%%%%%%%%%%%%%%%%%%%%%%%%%%%%%%%%%%%%%%%%%%%%%%%%%%%%%%%%%%%%%

% References are important to the reader; therefore, each citation must be complete and correct. If at all possible, references should be commonly available publications.

\bibliographystyle{IEEEtran}
\bibliographystyle{IEEEtran}
% \bibliography{root.bib}

\begin{thebibliography}{99}

\bibitem{le2014new}
S. Le Vine, M. Lee-Gosselin, A. Sivakumar, and J. Polak,
"A new approach to predict the market and impacts of round-trip and point-to-point carsharing systems: Case study of London,"
\textit{Transportation Research Part D: Transport and Environment}, vol. 32, pp. 218--229, 2014.

\bibitem{illgen2019literature}
S. Illgen and M. H{\"o}ck,
"Literature review of the vehicle relocation problem in one-way car sharing networks,"
\textit{Transportation Research Part B: Methodological}, vol. 120, pp. 193--204, 2019.

\bibitem{zardini2022analysis}
G. Zardini, N. Lanzetti, M. Pavone, and E. Frazzoli,
"Analysis and control of autonomous mobility-on-demand systems,"
\textit{Annual Review of Control, Robotics, and Autonomous Systems}, vol. 5, pp. 633--658, 2022.

\bibitem{boyaci2015optimization}
B. Boyac{\i}, K. G. Zografos, and N. Geroliminis,
"An optimization framework for the development of efficient one-way car-sharing systems,"
\textit{European Journal of Operational Research}, vol. 240, no. 3, pp. 718--733, 2015.

\bibitem{brar2022supply}
A. S. Brar, P. Kasture, and R. Su,
"Supply-Demand Balancing Model for EV Rental Fleet,"
in \textit{2022 IEEE 25th International Conference on Intelligent Transportation Systems (ITSC)}, pp. 1350--1355, 2022.

\bibitem{pavone2012robotic}
M. Pavone, S. L. Smith, E. Frazzoli, and D. Rus,
"Robotic load balancing for mobility-on-demand systems,"
\textit{The International Journal of Robotics Research}, vol. 31, no. 7, pp. 839--854, 2012.

\bibitem{brar2021dynamic}
A. S. Brar and R. Su,
"Dynamic Supply-Demand Balancing Policy for CMoD Fleet,"
in \textit{2021 IEEE International Intelligent Transportation Systems Conference (ITSC)}, pp. 2435--2440, 2021.

\bibitem{miao2015taxi}
F. Miao, S. Lin, S. Munir, J. A. Stankovic, H. Huang, D. Zhang, T. He, and G. J. Pappas,
"Taxi dispatch with real-time sensing data in metropolitan areas: A receding horizon control approach,"
in \textit{Proceedings of the ACM/IEEE Sixth International Conference on Cyber-Physical Systems}, pp. 100--109, 2015.

\bibitem{brar2020ensuring}
A. S. Brar and R. Su,
"Ensuring Service Fairness in Taxi Fleet Management,"
in \textit{2020 IEEE 23rd International Conference on Intelligent Transportation Systems (ITSC)}, pp. 1--6, 2020.

\bibitem{donovan2014new}
B. Donovan and D. B. Work,
"New york city taxi data (2010-2013),"
\textit{Dataset, http://dx.doi.org/10.13012/J8PN93H8}, 2014.

\bibitem{wang2023incentivized}
D. Wang and F. Liao,
"Incentivized user-based relocation strategies for moderating supply--demand dynamics in one-way car-sharing services,"
\textit{Transportation Research Part E: Logistics and Transportation Review}, vol. 171, p. 103017, 2023.

\bibitem{huang2020vehicle}
K. Huang, K. An, J. Rich, and W. Ma,
"Vehicle relocation in one-way station-based electric carsharing systems: A comparative study of operator-based and user-based methods,"
\textit{Transportation Research Part E: Logistics and Transportation Review}, vol. 142, p. 102081, 2020.

\bibitem{wang2021demand}
L. Wang, W. Ma, M. Wang, and X. Qu,
"Demand control model with combinatorial incentives and surcharges for one-way carsharing operation,"
\textit{Transportation Research Part C: Emerging Technologies}, vol. 125, p. 102999, 2021.

\end{thebibliography}

\end{document}